\documentclass{article}
\usepackage{amsthm,amsmath,amssymb,cite,fullpage}
\usepackage{times,cancel}
\usepackage{graphicx,xcolor}
\usepackage{hyperref}
\newtheorem{theorem}{Theorem}
\newtheorem{lemma}{Lemma}
\newtheorem{conj}{Conjecture}

\usepackage[utf8]{inputenc}
\usepackage{authblk}
\title{Optimal vaccination at high reproductive numbers: sharp transitions and counter-intuitive allocations}
\date{February 2022}
\author[a]{Nir Gavish\footnote{Corresponding author: ngavish@technion.ac.il}
}
\author[b]{Guy Katriel} 

\affil[a]{Faculty of Mathematics, Technion - IIT, Haifa 32000, Israel}
\affil[b]{Department of Applied Mathematics, ORT Braude College of Engineering, Karmiel 216100, Israel}

\begin{document}
\maketitle
\begin{abstract}
Optimization of vaccine allocations among different segments of a heterogeneous population is important for enhancing the effectiveness of vaccination campaigns in reducing the burden of epidemics. Intuitively, it would seem that allocations designed to minimize infections should prioritize those with the highest risk of being infected and infecting others. This prescription is well supported by vaccination theory, e.g., when the vaccination campaign aims to reach herd immunity. In this work, we show, however, that for vaccines providing partial protection (leaky vaccines) and for sufficiently high values of the basic reproduction number, intuition is overturned: the optimal allocation for minimizing the number of infections prioritizes the vaccination of those who are {\em least} likely to be infected. Furthermore, we show that this phenomenon occurs at a range of basic reproduction numbers relevant for the currently circulating strains of SARS-CoV-19. The work combines numerical investigations, asymptotic analysis for a general model, and complete mathematical analysis in a simple two-group model. The results point to important considerations in managing vaccination campaigns for infections with high transmissibility.
\end{abstract}

\section{Introduction}
Vaccination is the foremost tool for reducing the burden of contagious diseases. Mathematical models describing the spread of epidemics are widely used to explore the effect of vaccination on transmission dynamics and epidemic outcomes~\cite{ctx10278070460003971,fox1971herd,fine1993herd,anderson1992infectious}.   Heterogeneity of the population along various dimensions plays a vital role in shaping transmission dynamics. Different sub-populations, such as age groups or occupational groups, have different levels of contact among themselves and with others and may vary in their intrinsic susceptibility to being infected due to biological factors. Heterogeneity provides important opportunities for enhancing the effectiveness of vaccination campaigns by prioritizing the vaccination of some sub-populations  \cite{hethcote1987epidemiological}. This has led to a vast literature addressing the question of optimal vaccination in heterogeneous populations, under constraints on the number of vaccines that can be administered, with various aims, such as achieving herd immunity~\cite{ball2004stochastic,hill2003critical,duijzer2016most,Bai2021}, minimizing the initial growth rate of an emerging epidemic \cite{wallinga2010optimizing} or minimizing epidemic outcomes - the total number of infections, morbidity or mortality ~\cite{keeling2012optimal,kuddus2021mathematical,medlock2009optimizing,bansal2006comparative,yaari2016model}. 
These efforts have naturally received widespread attention in the wake of the COVID-19 pandemic~\cite{Matrajteabf1374,meehan2020age,Bubar916,islam2021evaluation,moore2021vaccination,Gavish2021.04.26.21256101}.  Yet, while most  studies  of optimal vaccination 
have involved transmission at relatively low reproduction numbers, e.g., associated with Influenza, often within the range of values for which herd immunity can be achieved, recent lineages of SARS-CoV-19 are associated with high reproduction numbers ($R_0$)~\cite{burki2021omicron,liu2021reproductive,ke2021estimating}. 
In this work, we show that high $R_0$ can lead to a dramatic and surprising change in the optimal allocation of vaccines.

We consider here the problem of optimal distribution of a {\em leaky} vaccine  - one providing partial protection - 
to a 
heterogeneous population, to minimize the total number of infections during an epidemic.

\begin{figure}[ht!] 
\begin{center}
   \includegraphics[width=0.75\textwidth] {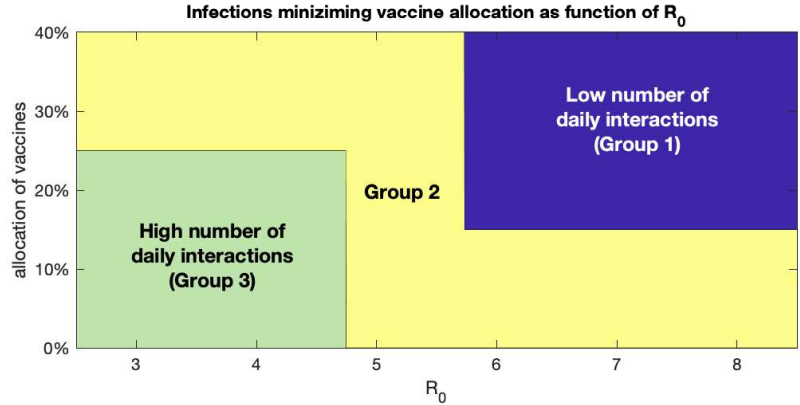} 
   \caption{{\bf Example with three groups.} Individuals of each of the groups interact with those of all other groups (random mixing), but members of group 2 (50\% of the population) have twice as many daily interactions as members of group 1  (25\% of the population), and members of group 3 (25\% of the population) have twice as many interactions as members of group 2. We consider the optimal allocation of vaccine to 40\% of the population, where vaccine efficacy in blocking transmission is 80\%.  Graph presents the optimal vaccine allocations as a function of~$R_0$.  At sufficienly high~$R_0$ it is optimal to vaccinate the entire
{\em low-contact} group, with the remaining vaccines allocated to the medium contact group so that the high contact group is left entirely non-vaccinated. }\label{fig:3group}
   \end{center}
\end{figure}
We first demonstrate this claim using a simple example that considers a population subdivided into three sub-populations differing by their number of daily contacts: low-contact, medium-contact, and high-contact groups, but with homogeneous mixing among all individuals.  
Figure \ref{fig:3group} shows the optimal allocation for minimizing the number of infections 
for a given vaccine coverage\footnote{See the figure's caption for details on the parameters in this example}, in dependence on the reproductive number $R_0$. As expected, for lower values of $R_0$ ($R_0<4.7$), the optimal allocation covers the entire high-contact group, with the remaining vaccines allocated to the medium-contact group.  However, for~$R_0>5.8$
it is optimal to vaccinate the entire
{\em low-contact} group, with the remaining vaccines allocated to the medium contact group so that the high contact group is left entirely non-vaccinated! This `paradoxical' effect, whereby it is sometimes optimal to vaccinate those less prone to infection while leaving those more prone to infection unprotected, is the focus of this work. 

At an intuitive level, the seemingly paradoxical optimal allocation of vaccines to those who are 
at lower risk of being infected can be understood as follows: for sufficiently high $R_0$, breakthrough infections among those who are more prone to infection become so common that it becomes inefficient to allocate vaccines to those groups. Instead, the optimal strategy is to move to a second line of defense and vaccinate sub-populations that are less prone to infection.

We further claim that the transition to infections-minimizing allocations which prioritize those who are {\em least} prone to infection occurs for values of the parameter~$R_0$ which are relevant to SARS-CoV-19. 
\begin{figure}[ht!]
\centering
   \includegraphics[width=0.75\textwidth]{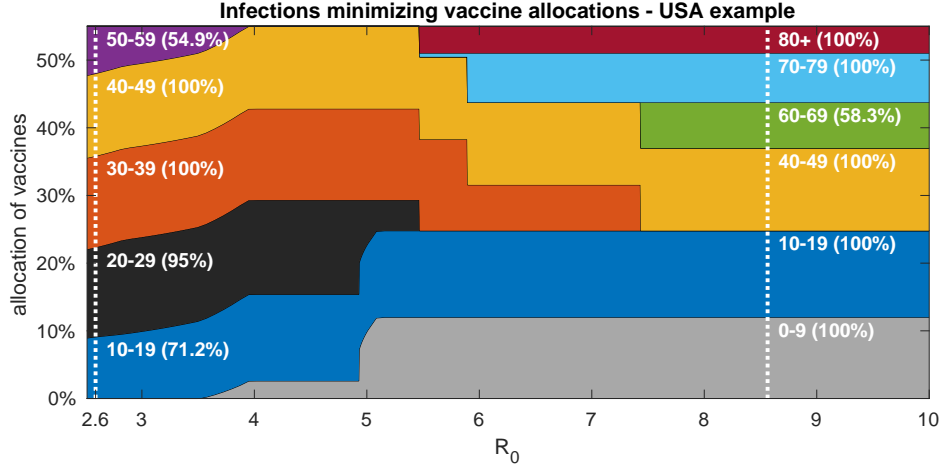} 
\caption{{\bf Example using the USA demographic and social structure.}
Optimal vaccine allocations for minimizing the overall number of infections for vaccine coverage of 55\% of the population with vaccine efficacy of~80\% in blocking transmission. For low values of~$R_0$, vaccination of spreaders is the best way to minimize the overall number of infections. However, for~$R_0>6$, it becomes optimal to allocate vaccines to the elderly, at the expense of vaccine allocation to those who contribute most to spreading.}
\label{fig:US_4group_example_coverage50}
\end{figure}
To demonstrate this point, we present a more realistic example computed using an age-structured transmission model with the estimated contact matrix for the USA. We assume that the susceptibility of children younger than $10$ is 50\% lower than that of older individuals, following estimates for SARS-Cov-19~\cite{davies2020age,dattner2021role,goldstein2021effect}. In this case, we observe that for values of~$R_0>6$, it becomes optimal to allocate vaccines to those of age 60 and older, despite their lower number of contacts, and to children under 10, despite their lower susceptibility, at the expense of vaccine allocation to those who contribute most to the spread of disease, see Figure~\ref{fig:US_4group_example_coverage50}. We have considered various scenarios in eight other countries and observed similar results.  

The analytical results presented in this work rigorously
characterize the phenomenon illustrated by the above examples and show that it is generic. Our key result considers the general case of a population divided into $n$ groups, with a general structure of 
contacts among the groups, and uses asymptotic analysis to obtain an explicit vaccine allocation which is nearly optimal for large $R_0$. The nearly-optimal allocation (i.e., optimal up to an exponentially small error) is to administer vaccines to those
who are {\it{least}} prone to infection, as determined by an {\it{exposure index}} which is computed for each group as a weighted sum of the number of interactions made by members of that group. 
In the example presented in Figure~\ref{fig:US_4group_example_coverage50}, the theoretical `nearly-optimal' allocation in fact coincides with the optimal allocation for~$R_0>7.5$.

To obtain a more refined theoretical 
understanding of the phenomena of interest,
we also consider a minimal case of two groups with different susceptibilities, which is amenable to complete mathematical analysis.  In this case, we prove the existence of a critical threshold $R_0^*$ at which a sharp transition occurs in the optimal allocation, from vaccinating those who are more susceptible to an allocation which prioritizes those who are less susceptible. 

\section{Vaccine allocation at high reproductive number}
We consider optimal allocation of vaccines, using a standard SIR model~\eqref{eq1}-\eqref{eq4}, in which the population is subdivided into $n$ groups, with~$N_j$ the fraction of population in group~$j$.  Thus, $\sum_{j=1}^n N_j=1$.

The asymptotic analysis of vaccine allocations at high~$R_0$ gives rise to a quantity~$e_j$ associated to each group~$j$, which we refer to as the {\it{exposure index}} of members of the group:
\begin{equation}\label{eq:exposureIndex}
e_j=\sigma_j \sum_{k=1}^n C_{jk}\eta_k.
\end{equation}
Here,~$C_{jk}$ is the mean number of contacts of a single member of group $j$ with members of group $k$ per unit time, $\sigma_j$ is the susceptibility of members of group $j$ relative to those of group $1$, and $\eta_j$ is the infectivity of 
members of group $j$ relative to those of group $1$ (thus $\sigma_1=\eta_1=1$). We remark that 
the quantities $e_j$ can be interpreted as proportional to the probabilities per unit time of an individual of group $j$ to be infected, assuming this individual is placed in a population all of whose members are infected.

We assume a vaccine provides imperfect leaky protection \cite{halloran2010design}, reducing susceptibility by a factor $\varepsilon\in (0,1)$ (thus $1-\varepsilon$ is the vaccine efficacy). Vaccination is performed before the introduction of the pathogen. Vaccine allocation is defined by a vector $\vec{v}$
whose components ~$0\le v_j\le N_j$, are the fractions of the population in groups $j=1,\cdots,n$ which are vaccinated.  
The final size of the epidemic is computed as the solution of a standard final size equation (FSE) system, see~\eqref{eq:FSE}. 
We control the reproductive number $R_0$ by varying the transmission parameter $\beta$, keeping all other parameters fixed,
and denote the final size by
$Z(R_0,\vec{v})$.
To find the allocation minimizing the final size, given 
the total coverage $v$, we solve
the optimization problem
\begin{equation}\label{opt1}{\mbox{minimize  }} Z(R_0,\vec{v}),\;\;\; 0\leq v_j\leq N_j,\;\; \sum_{j=1}^n v_j=v.
\end{equation}

The following result shows that for sufficiently high 
values of $R_0$, prioritizing the vaccination of the 
groups with {\it{lowest}} 
exposure index provides a 
nearly-optimal vaccine allocation.

\begin{theorem}\label{thm:main_asymptotic}   Assume, without loss of generality, that the groups are ordered in non-decreasing order of their exposure indices given by~\eqref{eq:exposureIndex},~$e_1\le e_2 \le\cdots \le e_n$.

Then, for given vaccine coverage~$v$, the vaccine allocation
\begin{equation}\label{eq:optimalAllocation}
\vec{v}^*=\left(N_1,N_2,\cdots,N_{j^*},v-\sum_{k=1}^{j^*}N_k,0,\cdots,0\right)
\end{equation}
where $j^*$ is the maximal index for which~$\sum_{k=1}^{j^*}N_j\leq v$, 
gives rise to a nearly minimal final size (overall number of infections)
in the sense that, defining $\vec{v}_{optimal}$ to be a minimizer of \eqref{opt1}, we have, for some $c>0$,
\begin{equation}\label{eq:error}
E=\frac{Z(R_0,\vec v^*)-Z(R_0,\vec v_{\rm optimal})}{1-Z(R_0,\vec v_{\rm optimal})}=O(e^{-c R_0}),\;\;\;{\mbox{as}}\;\;R_0\rightarrow\infty.
\end{equation}
\end{theorem}
Proof: See Section 5.3 and Appendix A. $\Box$

The example given in Figure~\ref{fig:US_4group_example_coverage50}
illustrates the results of Theorem~\ref{thm:main_asymptotic}, using parameters corresponding to the USA demography and contact structure~\cite{prem2020projecting} with nine 10-year age groups, and with an age-dependent susceptibility profile for SARS-CoV-19 in which the susceptibility of children is reduced~\cite{davies2020age}. Under these conditions, age group 0-9 has the smallest exposure index, followed by,  with increasing exposure index, age groups 70-79, 80+, 10-19, 40-49, 60-69, 50-59, and 20-29. Finally, the age group 30-39 has the largest exposure index. Accordingly, in this case,  by Theorem~\ref{thm:main_asymptotic}, for sufficiently high ~$R_0$, it is (essentially) optimal to first allocate all vaccines to age group 0-9, and then, based on the availability of supply, $v$, to extend coverage to age groups 70-79, 80+, 10-19, 40-49, 60-69, and so on. Figure~\ref{fig:US_4group_example_coverage50} presents the optimal infection minimizing vaccine allocations as a function of~$R_0$ in the case vaccine efficacy in preventing infections is 80\%. For~$R_0\gtrsim7.5$, we observe that the computed optimal allocation coincides with the asymptotic allocation~$\vec{v}^*$, defined in \eqref{eq:optimalAllocation} in Theorem~\ref{thm:main_asymptotic}.  Among the age groups included in this vaccine allocation, age group 60-69 has the highest exposure index. Therefore, by Theorem~\ref{thm:main_asymptotic}, a gradual decrease of the overall vaccine coverage will lead to a reduction in vaccines allocated to age group 60-69 until this age group is completely removed from the optimal allocation (data not shown).  

Figure~\ref{fig:USExample_attackrate} presents the attack rates corresponding to 
different allocations of vaccines, in dependence on $R_0$, for the same example as in
Figure~\ref{fig:US_4group_example_coverage50}. It is interesting to observe that although
the asymptotic allocation~$\vec{v}^*$ given by \eqref{eq:optimalAllocation}
coincides with the optimal allocation only for~$R_0\gtrsim7.5$, see Figure ~\ref{fig:US_4group_example_coverage50}, already for~$R_0>5.7$ the attack rate 
corresponding to the asymptotic allocation
is very close to the attack rate corresponding to the optimal allocation.
Indeed, for $R_0>5.7$
the relative error $E$ defined by~\eqref{eq:error} is less than 1\%. 
Furthermore, in this regime of~$R_0$,
the asymptotic allocation~$\vec{v}^*$ performs much better than common allocations considered in the literature such as the uniform allocation or allocation of vaccines to those which are most prone to become infected.
We have performed 
similar computations using contact 
matrices estimated for different countries, different levels of vaccine efficacy and coverage, finding that the minimal values 
of $R_0$ for which the relative error 
$E$ is less than $\%1$ fall in the range~$4<R_0<10$ with mean~$R_0\approx6$, see Appendix C.  We therefore conclude that the asymptotic allocation 
given by Theorem~\ref{thm:main_asymptotic} is nearly optimal for parameter values which are relevant to some lineages of SARS-CoV-19.
\begin{figure}[ht!]
\centering
   \includegraphics[width=0.75\textwidth]{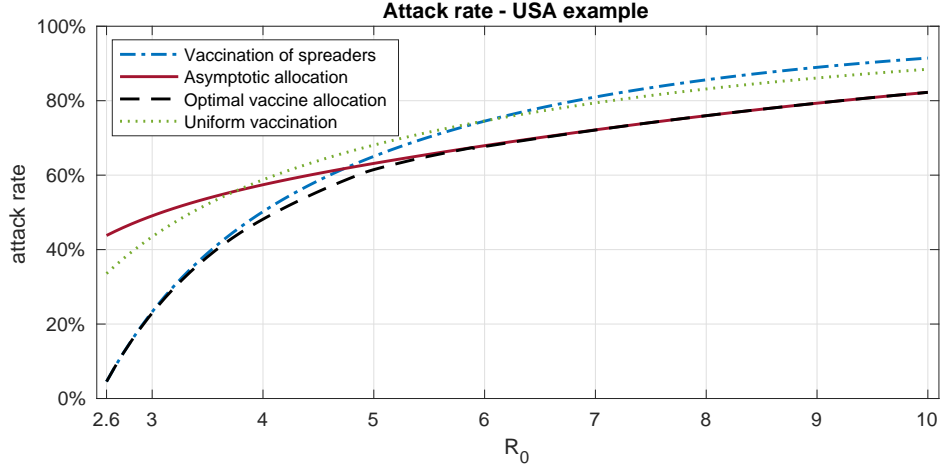}
\caption{{\bf Attack rate as a function of~$R_0$ for different allocations.}
The attack rate (percent of those infected among the population) as a function of~$R_0$ for the example presented in Figure~\ref{fig:US_4group_example_coverage50}. Dashed-dotted blue curve corresponds to the optimal allocation at~$R_0=2.6$, see Figure~\ref{fig:US_4group_example_coverage50}, which prioritizes those who contribute to the spread of the disease. The solid red curve corresponds to the asymptotic allocation~$\vec{v}^*$, see~\eqref{eq:optimalAllocation}, dashed black curve corresponds to the optimal vaccine allocation for minimizing infections, and dotted green curve corresponds to a uniform allocation of vaccines.}    
\label{fig:USExample_attackrate}
\end{figure}

\section{A two-group example}
\label{two-group}

In this section, we consider a minimal model 
allowing us to study how optimal 
vaccine allocations change with the 
reproductive number - the simplicity of the model enables us to obtain a complete mathematical analysis of the phenomena 
of interest. In this model, the population is composed of two
sub-populations (e.g., age groups) with fractions $N_1,\,N_2$ ($N_1+N_2=1$), which
differ only in their susceptibility to infection, with 
the susceptibility of group
$2$ to infection higher than that of group $1$ by a factor $\sigma>1$. 
Homogeneous
mixing among all individuals is assumed. Accordingly, the exposure index of group~$1$ equals~$e_1=1$, and it is smaller than the exposure index~$e_2=\sigma>1$. 
We assume a vaccine provides partial leaky protection with efficacy $1-\varepsilon$.

The dynamics of spread is described by the two-group SIR
model
\begin{eqnarray*}\label{s1}&&S_1'=-\beta S_1 I,\;\;\;V_1'=-\beta \varepsilon  V_1 I,\nonumber\\
&&S_2'=-\sigma\beta S_2 I,\;\;\;V_2'=-\sigma \beta \varepsilon  V_2 I\nonumber\\
&&I'=\beta [ S_1+\sigma  S_2+\varepsilon V_1+\sigma\varepsilon V_2]I-\gamma I,
\end{eqnarray*}
where $S_i(t),V_i(t)$ denote the size of the non-vaccinated and vaccinated parts of the population with lower ($i=1$)
and higher ($i=2$) susceptibility (exposure index), and $I(t)$ denotes the size of the population infected at time $t$, all as fractions of the total population, 
$\beta$ is the transmission rate, and $\gamma^{-1}$ is the 
duration of infectivity. 

We assume that the amount of available vaccine is sufficient to vaccinate a fraction $v$ of the population prior to the start of an epidemic.
We divide this fraction into two parts
$v_1,v_2$ with $v_1+v_2=v$, to be administered to 
each of the groups. For conciseness, we will assume below that 
$v<\min(N_1,N_2)$,
that is the amount available vaccine is insufficient to vaccinate any of the
two groups entirely - analogous results can be 
proved in the cases in which  it is possible to fully vaccinate one of the two groups.
Our initial conditions are thus
$$S_j(0)=N_j-v_j-I_j(0),\;\; V_j(0)=v_j,\;\;I(0)=I_1(0)+I_2(0).\;\; $$
The basic reproductive number (in the absence of vaccination) is $R_0=\frac{\beta}{\gamma}(N_1+\sigma N_2)$.

Our aim is to choose $v_1,v_2$ so as to minimize the fraction of population $Z=Z(v_1,v_2)$ which will be infected during the epidemic.
This fraction is given (in the limit of a small initial fraction of infected individuals, $I_1(0),I_2(0)\rightarrow 0$) as the solution $Z$ of the  final size equation 
\begin{equation}\label{fsz}1- (N_1-v_1)e^{- \frac{\beta}{\gamma}Z}-(N_2-v_2) e^{-\sigma \frac{\beta}{\gamma}Z}- v_1 e^{-\varepsilon \frac{\beta}{\gamma}Z}- v_2e^{-\sigma\varepsilon \frac{\beta}{\gamma}Z}  =Z,\end{equation}
and our optimal allocation problem is
\begin{equation}\label{op}
{\text{minimize}}\;Z(v_1,v_2)\qquad v_1,v_2\geq 0 ,\; v_1+v_2=v.
\end{equation}
Analyzing the above problem, we obtain 
\begin{theorem}\label{thm:main} Assuming $v<\min(N_1,N_2)$, there exists a value $R_0^*$ such that

(i) If $R_0<R_0^*$ then the optimal solution of \eqref{op} is given by
$(v_1,v_2)=(0,v)$.

(ii) If $R>R_0^*$ then the optimal solution of \eqref{op} is given by
$(v_1,v_2)=(v,0)$.

The transition value $R_0^*$ is given explicitly by the formula
$$R_0^*=\frac{(N_1+N_2\sigma)\ln(\frac{1}{\alpha})}{1- N_1\alpha-N_2\alpha^{\sigma}- v(\alpha^\varepsilon -\alpha)},$$
where $\alpha$ is the unique solution of the 
equation
$$\alpha-\alpha^{\sigma}-\alpha^{\varepsilon}+\alpha^{\sigma\varepsilon}=0,\;\;\;\alpha\in (0,1).$$
\end{theorem}
Proof: See Appendix B. $\Box$

Thus it is always optimal to allocate the entire stock of vaccines to 
only one of the groups, but there is a switch between vaccinating the 
{\it{more}} susceptible population, i.e., group 2, when $R_0$ is low to vaccinating the
{\it{less}} susceptible population, i.e., group 1, when $R_0$ is high.

A numerical example illustrating the contents of Theorem \ref{thm:main} is given in Figure~\ref{fig:TwoGroupExample}.
\begin{figure}[ht!] \label{two_group}
\begin{center}
    \includegraphics[width=0.75\textwidth]{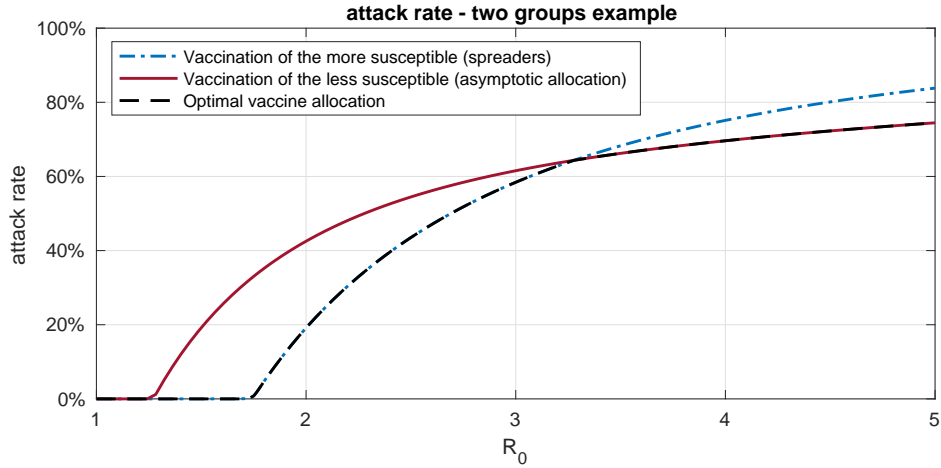} 
   \caption{{\bf Example of two groups with different susceptibilities.} The two groups are of equal size, one of which is 50\% as susceptible as the other. 40\% of the population can be vaccinated, with a vaccine that is 80\% effective in blocking transmission. The graph presents final sizes as a function of the basic reproduction number resulting from the allocation of vaccines to those of higher susceptibility (red curve), lower susceptibility (blue curve), and the optimal allocation (black dashed curve). }\label{fig:TwoGroupExample}
   \end{center}
\end{figure}
This figure reveals why the optimal vaccine allocation changes abruptly from allocations prioritizing the more susceptible to allocations prioritizing the less susceptible. Indeed, at low values of~$R_0$, the epidemic final size resulting from allocations prioritizing the more susceptible is lower than that resulting from vaccinating the less susceptible. However the final size curves intersect at the point~$R_0^{\rm crit}$, which is the point at which optimal allocation shifts to the more 
susceptible group. To prove Theorem \ref{thm:main}
one must also show that allocating some of the vaccine
to each of the groups is always inferior to one of the extreme allocations reserving
the entire stock of vaccines to only one group, see Appendix B.3.

\section{Discussion}

The results in this work reveal phenomena that have not been described in previous literature, to the best of our knowledge. These phenomena are relevant to circulating and plausibly to future lineages of SARS-CoV-19. They are also of theoretical importance, allowing us to refine our intuitions in thinking about the allocation of vaccines. Therefore, we believe the results should be of interest to researchers and public health experts involved in
vaccination policies.  

We have considered vaccination of a heterogeneous population, employing a vaccine that reduces the 
risk of infection (leaky vaccine), with the aim 
of minimizing the number of infections, and have shown that the optimal allocation of vaccines
displays a quite complicated dependence on the reproductive number $R_0$, with continuously varying rates in some intervals, constant allocation in others, and discontinuous shifts.
In particular, for sufficiently large
$R_0$ the optimal vaccine allocation has the 
counter-intuitive feature of favoring allocation of vaccines to the groups which are less prone
to infection, as measured by the exposure indices.
We have justified this result analytically, by
means of an asymptotic approximation which shows that
for sufficiently large values of $R_0$ the 
allocation of vaccines to the groups with
lower exposure indices result in a nearly-optimal allocation. Moreover, strongly supported by numerical results, we make the following
\begin{conj}
For sufficiently large 
$R_0$ the optimal allocation $\vec{v}_{\rm optimal}$ solving
\eqref{opt1} {\bf coincides}
with the asymptotic allocation~$\vec{v}^*$ given by \eqref{eq:optimalAllocation}.
\end{conj}
We have rigorously 
proved this conjecture in the two-group case considered in Theorem~\ref{thm:main}. 
The proof of this conjecture in the general case remains an interesting challenge.

We note that our results hinge on the leaky nature of vaccines. For example, in the explicit two-group example, the transition value~$R_0^*$, at which it becomes advantageous to vaccinate the less susceptible group,  goes to infinity in the limit $\varepsilon\rightarrow 0$ of perfect vaccine protection. Therefore, optimal allocations of perfect vaccines or of all-or-none vaccines will not undergo sharp transitions at high~$R_0$.  

It should be stressed that the phenomena which have been revealed are not merely theoretical curiosities, as our numerical examples show they occur for realistic parameter values, in particular for basic reproductive numbers, vaccine efficacy, and coverage values which are in the ranges relevant for the currently circulating strains of SARS-CoV-19.

A salient characteristic of the epidemiology of COVID-19 and other infectious diseases is the fact that the risk of severe disease and mortality is strongly 
age-dependent. In considering the allocation of vaccines, this fact has led to an apparent trade-off between minimizing mortality by direct protection of older age groups which are at higher risk of severe outcomes and minimizing the number of infections by vaccinating younger age groups who play a larger role in the spread of infection 
\cite{goldstein2021vaccinating,dushoff2021transmission,Goldsteine2107654118}. An interesting consequence of the results presented here is that, for sufficiently high
values of $R_0$, such a trade-off
does {\it{not}} exist, since the optimal allocation
for minimizing infections, covers the older age groups, {\it{because}} of the fact that they are less prone to infection. Under such conditions, prioritizing older age groups is optimal both for minimizing mortality {\it{and}}  for minimizing the number of infections.

While the focus here on prophylactic vaccine allocations that are optimal with respect to the overall number of infections, the methods, and the results can be modified to address optimization of other 
outcomes, such as severe cases and mortality, by weighing 
the infections in different age groups and according to vaccination status with appropriate rates of severe
disease and mortality. In a different direction, it is 
also of interest to consider measures such as the peak number of daily cases during the epidemic, as this measure relates directly to the load on healthcare providers.
We have found parameter regimes with high basic reproduction numbers for which vaccination of sub-populations that are least prone to infection gives rise to a lower peak number of cases compared to vaccine allocations prioritizing sub-populations that drive the epidemic spread. These results will be presented elsewhere.  

From a public health point of view, the sharp transition in optimal vaccine allocation to prioritization of sub-populations that are least prone to infection stems from the fact that breakthrough infections among those who are more prone to infection become so common that it becomes less effective to allocate vaccines to those groups. Therefore, the sharp transition reflects a withdrawal to a second line of defense.  We stress that an alternative policy to be considered is to apply non-pharmaceutical interventions to reduce the effective reproduction number to a level that enables avoiding such a withdrawal.  
More generally, it is important to be aware
of the fact that the most effective allocation of public health
resources to alleviating the burden of epidemics, e.g., quarantine policy, green passes, screening tests, and messaging to the public, may vary with the reproductive number, so that emergence of
a variant with higher transmissibility 
requires a re-evaluation of existing policies.

\subsubsection*{Data availability}
\small{All relevant data are within the manuscript, or on Zenodo at https://doi.org/10.5281/zenodo.5984404. The source code is available on GitHub at https://github.com/NGavish/HighR0.}

\subsubsection*{Acknowledgements}
The research of N.G. was supported by the ISRAEL SCIENCE FOUNDATION (grant No. 3730/20) within the KillCorona – Curbing Coronavirus Research Program.

\section{Materials and methods}
\subsection{Model}
We consider an SIR model in which the population is subdivided into $n$ groups. The dynamic variables $S_j$,$I_j$,$R_j$ and $V_j$ are the fractions of the total population consisting of
susceptible, infected, recovered, and vaccinated individuals in group $j$, respectively.  Model parameters are summarized in Table~\ref{tab:parameters}.\begin{table}
\centering
\caption{Model parameters}\label{tab:parameters}
\begin{tabular}{p{0.1\textwidth}p{0.68\textwidth}}
Parameter & Description\\
\hline
$N_j$ & Fraction of the population in  group $j$. \\
$C_{jk}$ & The mean number of contacts of a single member of group~$j$ with members of group $k$ per unit time. We denote by $C$ the $n\times n$  matrix with elements $C=\{C_{jk}\}_{j,k=1}^n$.\\
 $\sigma_j$ & The susceptibility of 
members of group $j$ relative to those of group $1$ (thus $\sigma_1=1$).\\
$\eta_j$ & The infectivity of 
members of group $j$ relative to those of group $1$ (thus $\eta_1=1$).\\
$\beta$ & Transmission parameter: the probability of infection of a non-vaccinated individual in group $1$ upon contact with a member of group $1$.\\
$\gamma$ & The recovery rate so that $\frac{1}{\gamma}$
is the mean duration of infectivity.\\
$v_j\in[0,1]$ & The fraction of the population in group $j$ which is vaccinated.\\
 $1-\varepsilon$ & Vaccine efficacy against infection, so that 
$\varepsilon$ is the factor by which the probability of infection upon contact is reduced for those vaccinated. \\
\hline
\end{tabular}
\end{table}

The dynamics is described by the differential equation system
\begin{equation}\label{eq1}S_j^\prime(t)=-\beta S_j(t)\sum_{k=1}^n D_{jk}\cdot I_k(t),\end{equation}
\begin{equation}\label{eq2}V_j^\prime(t)=-\varepsilon \beta V_j(t)\sum_{k=1}^n D_{jk}\cdot I_k(t),
\end{equation}
\begin{equation}\label{eq3}
I_j^\prime (t) =\beta (S_j(t)+\varepsilon V_j(t))\sum_{k=1}^n D_{jk}\cdot I_k(t)-\gamma I_j,
\end{equation}
\begin{equation}\label{eq4}
R_j'(t)=\gamma I_j(t),
\end{equation}
where
$$D_{jk}=\sigma_j \eta_k \cdot\frac{C_{jk}}{N_k}.$$

Since a proportion $v_j$ of group $j$
is vaccinated, we have 
\begin{equation}\label{ic}V_j(0)=v_j,\;\;S_j(0)=N_j-I_j(0)-R_j(0).\end{equation}

The basic reproductive number ~$R_0$ is equal to the spectral radius of ~$M=\frac{\beta}{\gamma}D$ \cite{diekmann,keeling}.

\subsection{Final size equations}
We will denote the attack rate by~$Z=\sum_{j=1}^m z_j$, where $z_j$ is the fraction of the total population which consisting of individuals of group $j$ who become infected. The quantities 
$z_j$ ($1\leq j\leq n$) are obtained as solutions of a system of final size equations
\begin{eqnarray}\label{eq:FSE}
N_j-z_j&=&(N_j-v_j)\exp\left[-\frac{\beta}{\gamma} \sum_{k=1}^n D_{jk}\cdot z_k\right]\nonumber\\&+&v_j\exp\left[-\varepsilon\frac{\beta}{\gamma} \sum_{k=1}^n D_{jk}\cdot z_k\right],\;\;\;1\leq j\leq n.
\end{eqnarray}

\subsection{Large $R_0$ asymptotics}
Our asymptotic analysis yields the following
\begin{lemma}\label{lem:asymptotic}
We have, for some $c>0$, 
$$\frac{1-Z(\beta,\vec{v})}{\sum_{j=1}^n v_j\exp\left[-\varepsilon\frac{\beta}{\gamma} \cdot e_j\right]}=1+O(e^{-c\beta})\;\;{\mbox{as }},\;\;\beta\rightarrow \infty,$$
where $e_j$ ($1\leq j\leq n$) are the 
exposure indices defined by \eqref{eq:exposureIndex}.
\end{lemma}
This approximation suggests the consideration of the following optimization problem
\begin{equation}\label{opt2}{\mbox{maximize}} \sum_{j=1}^n v_j\exp\left[-\varepsilon\frac{\beta}{\gamma} \cdot e_j\right],\quad 0\leq v_j\leq N_j,\quad \sum_{j=1}^n v_j=v,\end{equation}
as an approximation to the problem
\eqref{opt1}. 
The advantage of 
\eqref{opt2} is that it is a simple 
linear optimization problem. Indeed, it is immediate that to maximize \eqref{opt2} we need 
to allocate as much of the total vaccine available 
to those groups $j$ for which the coefficients are largest, which are precisely those for which the exposure indices $e_j$
is {\it{lowest}}, leading to the 
allocation \eqref{eq:optimalAllocation}.
Lemma \ref{lem:asymptotic} then leads to 
the conclusion that the allocation \eqref{eq:optimalAllocation} is indeed nearly optimal, as expressed by \eqref{eq:error}, see Appendix A.4.
\subsection{Numerical computation of optimal vaccine allocations}
We solve the optimization problem defined in~\eqref{opt1} using Matlabs' {\tt fmincon} nonlinear programming solver.   

The examples presented in the manuscript involve a  computation of solutions of the optimization problem for a range of basic reproduction numbers~$R_0$.  To do so,~$R_0$ is gradually increased, so that the initial guess used for basic reproduction number~$R_0+\delta R$ is the optimal allocation computed for the basic reproduction number~$R_0$.
To reduce the probability of convergence to a local minimum, we also repeatedly randomized the
initial allocation provided to the iterative  algorithm and verified that it converges to the same minimum so, that we are reasonably confident that we have found the global minima.

\subsection{Parameters used in examples}
In all example presented in the manuscript, we consider 80\% vaccine efficacy against infections, i.e.,~$\varepsilon=0.2$, and  a uniform infectivity profile~$\eta_i=1$ for all groups.
We note that the expressions for the reproductive number, as well as in the final size equations, depend on~$\beta/\gamma$.  Therefore we do not need to fix a value for the parameters~$\beta$ and~$\gamma$ separately, but rather adjust the value of~$\beta/\gamma$ to achieve the desired value of $R_0$.

The three group example presented in Figure~\ref{fig:3group} 
considers a case in which individuals of group $1$ (25\% of the population) interact half as frequently as those in group $2$, and individuals of group $3$ (another 25\% of the population) interact twice as frequently as those in group $2$.  Therefore, the group size parameters are~$N=(0.25,0.5,0.25)$, the corresponding contact matrix is 
\[
C=\left(
\begin{array}{ccc}
1&4&4\\
2&8&8\\
4&16&16
\end{array}
\right),
\]
and the susceptibility profile is~$\sigma_i=1$ for groups~$i=1,2,3$.

The example using the USA demography and social structure presented in Figure~\ref{fig:US_4group_example_coverage50} uses the following parameters:
\begin{itemize}
\item  {\bf{Age Demographics}} were taken from the UN World Population Prospects 2019 for each country~\cite{UNreport}, using $n=9$ age groups, of sizes
$N_j$ ($1\leq j\le 9$) corresponding to $10$-year 
increments, with the last group comprising those of age $80$ and older.
\item  {\bf{Contact matrix}} $C=\{C_{jk}\}_{j,k=1}^n$ was taken from~\cite{prem2020projecting}.  
Age bins in each case were originally provided in a 5-year increment, where the last age bin corresponds to ages 75 and older.  We follow the procedure as in~\cite{Bubar916} to adapt the matrices into 10-year increments.

\item {\bf{Relative susceptibility parameters}} $\sigma_i$ are taken from the age dependent relative susceptibility profile for SARS-CoV-19 from~\cite{davies2020age}:
$$(\sigma_1,\cdots,\sigma_9)=(1,0.95,  1.975,  2.15,2,2.05,2.2,1.85,1.85),$$
in which the relative susceptibility of age group 0-19 is roughly half those of older age groups. 

\end{itemize}

The parameters used in the two group example presented in Figure~\ref{fig:TwoGroupExample} are
\[
N_1=N_2=0.5,\qquad \sigma=2.
\]
\appendix
\section{Vaccination at high $R_0$: proofs}






\subsection{Final size formula}

The final-size formula yields the overall number of infections in each group in terms of the model parameters~\cite{diekmann,keeling}. 
We derive the final size formula corresponding to the above model.

From \eqref{eq1} and \eqref{eq2} we have
\[
\frac{S_j^\prime (t)}{S_j(t)}=-\beta\sum_{k=1}^n D_{jk}\cdot I_k(t) ,\qquad\frac{V_j^\prime(t)}{V_j(t)}=-\beta\varepsilon\sum_{k=1}^n D_{jk}\cdot I_k(t),
\]
which upon integration yields
\[
\log \frac{S_j(\infty)}{S_j(0)} = -\beta \sum_{k=1}^n D_{jk}\cdot \int_{0}^
\infty I_k(s)ds,\quad
\log \frac{V_j(\infty)}{V_j(0)} =  -\beta\varepsilon \sum_{k=1}^n D_{jk}\cdot\int_{0}^
\infty I_k(s)ds,
\]
or
\begin{equation}\label{int}
S_j(\infty)=S_j(0)\exp\left[-\beta \sum_{k=1}^n D_{jk}\cdot\int_{0}^
\infty I_k(s)ds\right], V_j(\infty)=V_j(0)\exp\left[-\beta\varepsilon \sum_{k=1}^n D_{jk}\cdot\int_{0}^
\infty I_k(s)ds\right].
\end{equation}
Summing \eqref{eq1}-\eqref{eq3}, we have
\[
I_j^\prime(t)=-S_j^\prime(t)-V_j^\prime(t)-\gamma I_j(t),
\]
which, upon integration, gives
\[
\cancel{I_j(\infty)}-I_j(0)=S_j(0)-S_j(\infty)+V_j(0)-V_j(\infty)-\gamma \int_{0}^\infty I_j(t),
\]
or
\begin{equation}\label{eq:int1}
\int_{0}^\infty I_j(t)=\frac1\gamma\left[S_j(0)-S_j(\infty)+V_j(0)-V_j(\infty)+I_j(0)\right].
\end{equation}
Since
$S_j(t)+V_j(t)+I_j(t)+R_j(t)=N_j$ for all $t$, and
$I_j(\infty)=0$, we have
\begin{equation}\label{eq:cons}
S_j(0)+V_j(0)+I_j(0)=N_j-R_j(0),\qquad S_j(\infty)+V_j(\infty)=N_j-R_j(\infty),
\end{equation}
and can write \eqref{eq:int1} as
\[
\int_{0}^\infty I_j(t)=\frac1\gamma\left[R_j(\infty)-R_j(0)\right],
\]
so that \eqref{int} yields
\begin{subequations}\label{eq:finalSizesystem}
\begin{equation}\label{eq:Sinf}
S_j(\infty)=S_j(0)\exp\left[-\frac{\beta}{\gamma} \sum_{k=1}^n D_{jk}\cdot \left[R_k(\infty)-R_k(0)\right]\right],
\end{equation}
\begin{equation}\label{eq:Vinf}
V_j(\infty)=V_j(0)\exp\left[-\varepsilon\frac{\beta}{\gamma} \sum_{k=1}^n D_{jk}\cdot\left[R_k(\infty)-R_k(0)\right]\right].
\end{equation}
\end{subequations}
Combining~\eqref{eq:cons} and \eqref{eq:finalSizesystem} yields
\[
N_j-R_j(\infty)=S_j(0)\exp\left[-\frac{\beta}{\gamma} \sum_{k=1}^n D_{jk}\cdot\left[R_k(\infty)-R_k(0)\right]\right]+V_j(0)\exp\left[-\varepsilon\frac{\beta}{\gamma} \sum_{k=1}^n D_{jk}\cdot\left[R_k(\infty)-R_k(0)\right]\right],
\]
or, defining $z_j=R_j(\infty)-R_j(0)$ to be fraction of the population consisting of individuals of 
group $j$ infected 
throughout the post-vaccination period:
\begin{equation}\label{ff}
N_j-R_j(0)-z_j=S_j(0)\exp\left[-\frac{\beta}{\gamma} \sum_{k=1}^n D_{jk}\cdot z_k\right]+V_j(0)\exp\left[-\varepsilon\frac{\beta}{\gamma} \sum_{k=1}^n D_{jk}\cdot z_k\right],\;\;\;1\leq j\leq n.
\end{equation}
Assuming that $R_j(0)=0$ (no unvaccinated individuals are initially immune), $I_j(0)=0$ (the initial fraction of infected individuals is very small) and setting $V_j(0)=v_j$, we 
have
\begin{equation}\label{ff1}
N_j-z_j=(N_j-v_j)\exp\left[-\frac{\beta}{\gamma} \sum_{k=1}^n D_{jk}\cdot z_k\right]+v_j\exp\left[-\varepsilon\frac{\beta}{\gamma} \sum_{k=1}^n D_{jk}\cdot z_k\right],\;\;\;1\leq j\leq n.
\end{equation}
The total number fraction of the population infected is then
$$Z=\sum_{j=1}^n z_j.$$

We note that in the 
particular case of homogeneous 
mixing and identical infectivity of all
groups, we have 
$$C_{jk}=\sigma_j cN_k,$$
where $c$ is the average number of contacts per individual, and summing the $n$ equations \eqref{ff}, we obtain
\begin{equation}\label{gfs}1-Z=\sum_{j=1}^n \left[(N_j-v_j)\exp\left[-\frac{c\beta\sigma_j}{\gamma} Z \right]+v_j\exp\left[-\varepsilon\frac{c\beta\sigma_j}{\gamma} Z \right]\right],
\end{equation}
This gives in 
particular the final-size equation \eqref{fsz} for two groups with varying susceptibility, used in section \ref{proofs2}, where $c$ is absorbed into $\beta$.

\subsection{Asymptotics of the final size 
for high transmission rate}
\label{asymptot}
We consider the final size $Z=\sum_{j=1}^n z_j$, 
where $z_j$ are given as the solutions of 
system of final size equations given by~\eqref{ff1}.
Although our main interest is to study behavior of~\eqref{ff1} in the regime of large basic reproduction numbers, the final size equations are not unique functions of~$R_0$.  Instead, we choose to study the asymptotics of the final size for high transmission rate~$\beta$ since the transmission rate is proportional to the basic reproduction number and since the dependence of~\eqref{ff1} on~$\beta$ is explicit.
To emphasize the dependence of the final size on 
the transmission rate $\beta$ and the vaccine
allocation $\vec{v}=(v_1,\cdots,v_n)$, fixing all
other parameters, we write $Z=Z(\beta,\vec{v})$.  We denote the set of possible allocations by
\begin{equation}\label{V}{\cal{V}}=\Big\{ \vec{v}\;\Big|\; 0\leq v_j\leq N_j,\;\; \sum_{j=1}^n v_j=v\Big\}.
\end{equation}

The following asymptotic result (Lemma 1 in the main text) is the 
key to obtaining an approximately optimal 
allocation for large $R_0$.
\begin{lemma}\label{lem:SIasymptotic}
We have, for some $c>0$, 
\begin{equation}\label{asymptotic}\frac{1-Z(\beta,\vec{v})}{\sum_{j=1}^n v_j\exp\left[-\varepsilon\frac{\beta}{\gamma} \cdot e_j\right]}=1+O(e^{-c\beta})\;\;{\mbox{as }}\;\;\beta\rightarrow \infty,\end{equation}
uniformly in $\vec{v}\in {\cal{V}}$,
where $e_j$ ($1\leq j\leq n$) are the 
exposure indices defined by \begin{equation}\label{eq:SIexposureIndex}
e_j=\sigma_j \sum_{k=1}^n C_{jk}\eta_k.
\end{equation}
\end{lemma}

To prove this lemma
we consider the solution
$z_j=z_j(\beta)$ ($1\leq j\leq n$) of the system 
\eqref{ff1} as functions of the 
transmission parameter $\beta$, fixing all other parameters:
\begin{equation}\label{ff2}
N_j-z_j(\beta)=(N_j-v_j)\exp\left[-\frac{\beta}{\gamma} \sum_{k=1}^n D_{jk}\cdot z_k(\beta)\right]+v_j\exp\left[-\varepsilon\frac{\beta}{\gamma} \sum_{k=1}^n D_{jk}\cdot z_k(\beta)\right],\;\;\;1\leq j\leq n.
\end{equation}
Since $z_k(\beta)\rightarrow N_j$ 
as $\beta\rightarrow \infty$, the 
right-hand side is exponentially small
in $\beta$, 
hence, setting
$$w_j(\beta)=\beta(N_j-z_j(\beta))$$
we have, for some $c>0$,
\begin{equation}\label{ll}w_j(\beta)=O(e^{-c\beta})\;\;\;{\mbox{as }}\;\;\beta\rightarrow \infty.\end{equation}
We re-write \eqref{ff2} in the form
\begin{eqnarray}\label{ff3}
w_j(\beta)&=&\beta(N_j-v_j)\exp\left[-\frac{\beta}{\gamma} \sum_{k=1}^n N_k D_{jk}\right]\exp\left[\frac{1}{\gamma} \sum_{k=1}^n D_{jk}\cdot w_k(\beta)\right]\nonumber\\&+&\beta v_j\exp\left[-\varepsilon\frac{\beta}{\gamma} \sum_{k=1}^n N_k D_{jk}\right]\exp\left[-\varepsilon\frac{1}{\gamma} \sum_{k=1}^n D_{jk}\cdot w_k(\beta)\right],\;\;\;1\leq j\leq n.
\end{eqnarray}
By \eqref{ll},
\eqref{ff3} we have, as $\beta\rightarrow \infty$,
\begin{eqnarray}\label{ff4}
w_j(\beta)&=&\left(\beta(N_j-v_j)\exp\left[-\frac{\beta}{\gamma} \sum_{k=1}^n N_k D_{jk}\right]+\beta v_j\exp\left[-\varepsilon\frac{\beta}{\gamma} \sum_{k=1}^n N_kD_{jk}\right]\right)(1+O(e^{-c\beta}))\nonumber\\
&=&
\beta v_j\exp\left[-\varepsilon\frac{\beta}{\gamma} \sum_{k=1}^n N_k D_{jk}\right](1+O(e^{-c\beta})),\;\;\;1\leq j\leq n.
\end{eqnarray}
Therefore,
$$z_j(\beta)=N_j-\frac{w_j(\beta)}{\beta}=N_j- v_j\exp\left[-\varepsilon\frac{\beta}{\gamma} \sum_{k=1}^n N_k D_{jk}\right](1+O(e^{-c\beta})),$$
so that 
$$1-Z(\beta)=\sum_{j=1}^n (N_j-z_j(\beta))=(1+O(e^{-c\beta})) \sum_{j=1}^n v_j\exp\left[-\varepsilon\frac{\beta}{\gamma} \sum_{k=1}^n N_k D_{jk}\right]$$
which is the result of lemma
\ref{lem:SIasymptotic}.

\subsection{The approximately optimal 
allocation for large $R_0$}

Our optimization problem is to minimize the 
number infections by allocating vaccines to a given fraction
$v$ of the population:
\begin{equation}\label{SIopt1}{\mbox{minimize  }} Z(\beta,\vec{v}),\;\;\;\vec{v}\in {\cal{V}}.
\end{equation}
where ${\cal{V}}$ is given by \eqref{V}.

Motivated by Lemma \ref{lem:SIasymptotic}, we 
introduce the approximate optimization problem
\begin{equation}\label{SIopt2}{\mbox{maximize }} \Psi(\beta,\vec{v})=\sum_{j=1}^n v_j\exp\left[-\varepsilon\frac{\beta}{\gamma} \cdot e_j\right],\;\;\;\vec{v}\in {\cal{V}}.\end{equation}
This is a linear optimization problem whose 
solution is immediate: 
assuming, without loss of generality, that the groups are ordered in non-decreasing order of their exposure index
\[
e_1\le e_2 \le\cdots \le e_n,
\]
the solution $\vec{v}^*$ of \eqref{SIopt2} prioritises
groups with low exposure index, that is 
\begin{equation}\label{eq:SIoptimalAllocation}
\vec{v}^*=\left(N_1,N_2,\cdots,N_{j^*},v-\sum_{k=1}^{j^*}N_k,0,\cdots,0\right)
\end{equation}
where $j^*$ is the maximal index for which~$\sum_{k=1}^{j^*}N_j\leq v$.

Theorem 1 in the main text justifies the
claim that $v^*$ is a nearly-optimal 
solution of \eqref{SIopt1}. We now prove this theorem.

Assuming $\vec{v}_{\rm optimal}$ is 
a minimizer of \eqref{SIopt1}, we have, substituting
$\vec{v}=\vec{v}_{\rm optimal}$ into 
\eqref{asymptotic}, and using the
fact that $\vec{v}^*$ is the maximizer of
\eqref{SIopt2},
\begin{equation}\label{ee1}1-Z(\beta,\vec{v}_{\rm optimal})=\left(1+O(e^{-c\beta})\right)\sum_{j=1}^n v_{{\rm optimal},j}\exp\left[-\varepsilon\frac{\beta}{\gamma} \cdot e_j\right]\leq \left(1+O(e^{-c\beta})\right)\sum_{j=1}^n v_{j}^*\exp\left[-\varepsilon\frac{\beta}{\gamma} \cdot e_j\right].\end{equation}
Substituting $\vec{v}=\vec{v}^*$ into
\eqref{asymptotic} we have
\begin{equation}\label{ee2}1-Z(\beta,\vec{v}^*)=\left(1+O(e^{-c\beta})\right)\sum_{j=1}^n v_j^*\exp\left[-\varepsilon\frac{\beta}{\gamma} \cdot e_j\right]\;\;\Rightarrow\;\;\sum_{j=1}^n v_j^*\exp\left[-\varepsilon\frac{\beta}{\gamma} \cdot e_j\right]=\left(1+O(e^{-c\beta})\right)[1-Z(\beta,\vec{v}^*)].\end{equation}
Combining \eqref{ee1}, \eqref{ee2} we have
\begin{equation}\label{ee3}1-Z(\beta,\vec{v}_{\rm optimal})\leq \left(1+O(e^{-c\beta})\right)[1-Z(\beta,\vec{v}^*)].\end{equation}
On the other hand, by the definition of 
$\vec{v}_{\rm optimal}$ we certainly have
$$1-Z(\beta,\vec{v}_{\rm optimal})\geq 1-Z(\beta,\vec{v}^*),$$
which combined with \eqref{ee3}, yields
\begin{equation}\label{ee4}
\frac{1-Z(\beta,\vec{v}_{\rm optimal})}{1-Z(\beta,\vec{v}^*)}=1+O(e^{-c\beta})\;\;\Rightarrow\;\;
\frac{1-Z(\beta,\vec{v}^*)}{1-Z(\beta,\vec{v}_{\rm optimal})}=1+O(e^{-c\beta}),
\end{equation}
which is equivalent to Eq. 4 in Theorem 1 of the main text.

\section{The two-group example: proofs}

\label{proofs2}

\subsection{The two-group model}
In this section we consider the two-group example 
\begin{eqnarray*}\label{SIs1}&&S_1'=-\beta S_1 I,\;\;\;V_1'=-\beta \varepsilon  V_1 I,\nonumber\\
&&S_2'=-\sigma\beta S_2 I,\;\;\;V_2'=-\sigma \beta \varepsilon  V_2 I\nonumber\\
&&I'=\beta [ S_1+\sigma  S_2+\varepsilon V_1+\sigma\varepsilon V_2]I-\gamma I
\end{eqnarray*}
with $\sigma>1$ (group $2$ is more susceptible), and initial conditions 
$$S_j(0)=N_j-v_j-I_j(0),\;\; V_j(0)=v_j,\;\;I(0)=I_1(0)+I_2(0),$$
where~$N_1+N_2=1.$

Setting
\begin{equation}\label{dr1}r=\frac{\beta}{\gamma},\end{equation} 
the reproductive number corresponding to this model (without vaccination) is
$$R_0=r(N_1+\sigma N_2).$$
The post-vaccination reproduction number is then
$$R_{\rm vac}=\frac{\beta}{\gamma}\left[(N_1-v_1)+(N_2-v_2)\sigma+v_1\varepsilon +v_2\sigma\varepsilon \right]=R_0\left(1-\frac{\left(v_1 +v_2\sigma\right)(1-\varepsilon)}{N_1+\sigma N_2}\right),$$
so that herd immunity can be achieved and the epidemic can be prevented if $R_{\rm vac}\leq 1$.
It is immediate from the above expression that the lowest value of $R_{\rm vac}$ will be achieved if all vaccines are allocated to the more susceptible group: $v_1=0,v_2=v$,
giving $R_{\rm vac}=R_0\left(1-\frac{\cdot v\sigma (1-\varepsilon)}{N_1+\sigma N_2}\right)$,
so that herd immunity will be achieved if 
\begin{equation}\label{crit}R_0\left(1-\frac{ v\sigma(1-\varepsilon)}{N_1+N_2\sigma}\right)\leq 1\;\;\Leftrightarrow\;\;R_0\leq R_0^{crit}\doteq\left(1-\frac{ v\sigma(1-\varepsilon)}{N_1+N_2\sigma}\right)^{-1}. \end{equation}
In case herd immunity cannot be achieved, our aim is to chose $v_1,v_2$ so as to minimize the fraction of population $Z=Z(v_1,v_2)$ which will be infected during the epidemic.
This fraction is given by the final size equation
\begin{equation}\label{SIfsz}1- (N_1-v_1)e^{- rZ}-(N_2-v_2) e^{-\sigma rZ}- v_1 e^{-\varepsilon rZ}- v_2e^{-\sigma\varepsilon rZ}  =Z,\end{equation}
This is a special case of the final size equation for $n$ groups, see \eqref{gfs}.

We consider the problem of minimizing the final size,
\begin{eqnarray}\label{SIop}&&{\text{minimize}}\;Z(v_1,v_2)\\ &&{\text{subject to}}\;\;v_1,v_2\geq 0 ,\;v_1+v_2=v.\nonumber\;\;\;\end{eqnarray}
In this case we can perform a complete 
analysis of the problem, to show how
the optimal allocation shifts from
the more susceptible to the less susceptible group as $R_0$ crosses 
a threshold, as described in (Theorem 2 of the main text):
\begin{theorem}\label{main} Assuming $v<\min(N_1,N_2)$, there exists a value $R_0^*$ such that

(i) If $R_0<R_0^*$ then the optimal solution of \eqref{SIop} is given by
$(v_1,v_2)=(0,v)$.

(ii) If $R>R_0^*$ then the optimal solution of \eqref{SIop} is given by
$(v_1,v_2)=(v,0)$.

The transition value $R_0^*$ is given explicitly by the formula
\begin{equation}\label{eq:R0*}
    R_0^*=\frac{(N_1+N_2\sigma)\ln(\frac{1}{\alpha})}{1- N_1\alpha-N_2\alpha^{\sigma}- v(\alpha^\varepsilon -\alpha)},
\end{equation}
where $\alpha$ is the unique solution of the 
equation
\begin{equation}\label{xse}\alpha-\alpha^{\sigma}-\alpha^{\varepsilon}+\alpha^{\sigma\varepsilon}=0,\;\;\;\alpha\in (0,1).\end{equation}
\end{theorem}
This section is devoted to the proof of this theorem \ref{main}.

\subsection{Comparing single-group vaccination schedules}

We will first compare the two options in which the entire supply of vaccine is allocated to 
one of the groups, showing that allocating all the vaccine to the more
susceptible group is the superior option when $R_0<R_0^*$ and the inferior option when $R_0>R_0^*$. Later, in Section \ref{mixed}, we will show that no mixed option is better 
than these two extreme options.

For $i=1,2$, let $Z_i(r)$ be the attack rate in case all the vaccine
is allocated to group $i$, in dependence on $r$, and 
assuming all other parameters are fixed. Then, using \eqref{SIfsz} with
$(v_1,v_2)=(v,0)$ and $(v_1,v_2)=(0,v)$, respectively, we have
\begin{equation}\label{f1}f_1(Z_1(r))  =Z_1(r),\end{equation}
\begin{equation}\label{f2}f_2(Z_2(r))  =Z_2(r),\end{equation}
where
$$f_1(Z)=1- (N_1-v)e^{- rZ}-N_2e^{-\sigma rZ}- v e^{-\varepsilon rZ},$$
$$f_2(Z)=1- N_1e^{- rZ}-(N_2-v)e^{-\sigma rZ}- ve^{-\sigma\varepsilon rZ}.$$

We note some basic properties of the functions $f_1,f_2$, to be used below.
\begin{lemma}\label{bp}
(i) \eqref{f1} has a unique nontrivial solution $Z_1(r)\in (0,1)$ iff $r>r_1= \frac{1}{ (N_1-v)+\sigma N_2+\varepsilon v}$.

(ii) \eqref{f2} has a unique nontrivial solution $Z_2(r)\in (0,1)$ iff $r>r_2= \frac{1}{ N_1+\sigma (N_2-v)+\sigma \varepsilon v}$.

(iii) For $i=1,2$: when $r>r_i$ we have
\begin{eqnarray}\label{ine}
0<Z<Z_i(r)\;\;&\Leftrightarrow&\;\; Z<f_i(Z)\\
Z_i(r)<Z\leq 1\;\;&\Leftrightarrow&\;\; Z>f_i(Z)
\end{eqnarray}

(iv) $r_1<r_2$.
\end{lemma}

\begin{proof} 
Define $h_i(Z)=f_i(Z)-Z$ ($i=1,2$).
$h_i(Z)$ is concave and satisfies $h_i(0)=0$, $h_i(1)<0$, so that $h_i(Z)=0$ has a unique solution $Z_i(r)\in (0,1)$ if and only if $h_i'(0)>0$, which is equivalent to the condition $r>r_i$.
To obtain (iii) note that when $r>r_i$,  $h_i(Z)$ is positive for $Z>0$ small, while $h_i(1)<0$. Since the only point in
$(0,1)$ at which $h_i$ vanishes is $Z=Z_i(r)$ we conclude that $h_i(Z)>0$ for $Z\in (0,Z_1(r))$ and $h_i(Z)<0$ for $Z\in (Z_1(r),1)$, which gives \eqref{ine}.
Finally, (iv) follows by direct calculation, using the definitions of $r_1,r_2$ and the assumption that $\sigma>1$.
\end{proof}

Note that when $r\leq r_i$, we have
$Z_i(r)=0$, that is an epidemic is averted by allocating the 
vaccine to group $i$. Thus if $r\leq r_1$ then
$Z_1(r)=Z_2(r)=0$, hence herd immunity 
can be reached by allocating the vaccine to either one of the
two groups, and when
$r_1<r\leq r_2$ we have $Z_1(r)>0,Z_2(r)=0$,
hence for $r$ in this range one should certainly 
vaccinate the more susceptible group $2$, and thus attain herd immunity. We are interested in the case 
$r>r_2$, which is equivalent to $R_0>R_0^{crit}$ (see \eqref{crit} above), so this will condition will be assumed below.

\begin{lemma}\label{c2}
Assume $r>r_2$.
Define $\kappa=\kappa(\sigma,\varepsilon)$ by $\kappa=\ln\left(\frac{1}{\alpha}\right)$
where $\alpha$ is the (unique) solution of  \eqref{xse}.
Then 
$f_1(Z)>f_2(Z)$ when $0<Z<\frac{\kappa}{r}$ and $f_2(Z)>f_1(Z)$ when $Z>\frac{\kappa}{r}$.
\end{lemma}
\begin{proof}

Define
\begin{equation}\label{defg}g(x)= x-x^{\sigma}-x^{\varepsilon}+x^{\sigma\varepsilon}.
\end{equation}
We first prove that the equation $g(x)=0$
has a unique solution $\alpha\in (0,1)$, and
that
\begin{equation}\label{pg}
0<x<\alpha\;\;\Rightarrow\;\; g(\alpha)<0, \;\;\;
\alpha<x<1\;\;\Rightarrow \;\; g(x)>0.
\end{equation}
We have $g(0)=g(1)=0$ and,
since 
$$\lim_{x\rightarrow 0+}\frac{g(x)}{x} =\lim_{x\rightarrow 0+}\left[(1-x^{\sigma-1})+x^{\varepsilon-1}(x^{(\sigma-1)\varepsilon}-1)\right]=-\infty$$
we have $g(x)<0$ for $x>0$ sufficiently small.
$$g'(x)= 1-\sigma x^{\sigma-1}-\varepsilon x^{\varepsilon-1}+\sigma\varepsilon x^{\sigma\varepsilon-1}\;\;\Rightarrow\;\;g'(1)=(1-\sigma)(1-\varepsilon)<0$$
implies $g(x)>0$ for $x<1$ sufficiently close to $1$.
Therefore there 
exists $\alpha\in (0,1)$ with $g(\alpha)=0$, that is a solution of \eqref{xse}. 
We now show that this solution is 
unique, and hence that \eqref{pg} holds.
Note that \eqref{xse} is equivalent to
$$h(x)=x^{\sigma-1}+x^{\varepsilon-1} -x^{\sigma\varepsilon-1}=1.$$
Assume by way of contradiction that
$h(x)=1$ 
has two solutions in $(0,1)$, 
in addition to the solution $x=1$.
Then, by Rolle's theorem, the function
$$h'(x)=(\sigma-1) x^{\sigma-2}+(\varepsilon-1)x^{\varepsilon-2} -(\sigma\varepsilon-1)x^{\sigma\varepsilon-2}$$
has at least two zeros in $(0,1)$, and 
multiplying both sides by $x^{2-\sigma\varepsilon}$, the function
$$x^{2-\sigma\varepsilon}h'(x)=(\sigma-1) +(\varepsilon-1)x^{\varepsilon-\sigma} -(\sigma\varepsilon-1)x^{\sigma(\varepsilon-1)}$$
has at least two zeros in $(0,1)$,
hence, again by Rolle's theorem, the function
$$[x^{2-\sigma\varepsilon}h'(x)]'= (\varepsilon-1)(\varepsilon-\sigma)x^{\varepsilon-\sigma-1} -(\sigma\varepsilon-1)\sigma(\varepsilon-1)x^{\sigma(\varepsilon-1)-1}$$
has at least one zero $x\in (0,1)$. 
But this implies
$$x^{(\sigma-1)\varepsilon}=\frac{\sigma-\varepsilon}{(1-\sigma\varepsilon)\sigma}, $$
hence, in view of the assumptions $\sigma>1$,
$\varepsilon\in (0,1)$,
$$0<\frac{\sigma-\varepsilon}{(1-\sigma\varepsilon)\sigma} <1\;\;\Rightarrow\;\;\sigma-\varepsilon <(1-\sigma\varepsilon)\sigma\Rightarrow\;\;
\sigma^2<1,$$
contradicting the assumption $\sigma>1$.

Noting now that we have
\begin{equation*}\label{aa}f_1(Z)>f_2(Z)\;\;\Leftrightarrow\;\;
e^{- rZ}
 -e^{-\sigma rZ}- e^{-\varepsilon rZ}+ e^{-\sigma\varepsilon rZ}<0\;\;\Leftrightarrow\;\;
 g\left(e^{- rZ}\right)>0,
\end{equation*}
\eqref{pg} implies
$$rZ<\kappa=-\ln\left(\alpha\right) \;\;\Rightarrow\;\; e^{-rZ}>\alpha\;\;\Rightarrow\;\;g(e^{-rZ})>0\;\;\Rightarrow\;\; f_1(Z)>f_2(Z),$$
and similarly $rZ>\kappa$ implies 
$f_1(Z)<f_2(Z)$.
\end{proof}

\begin{lemma}\label{ad}
Let $\alpha$ be the unique solution of \eqref{xse},
$\kappa=\ln\left(\frac{1}{\alpha}\right)$, and
\begin{equation}\label{drs}r^*=\frac{\kappa}{1- (N_1-v)\alpha-N_2\alpha^\sigma- v \alpha^\varepsilon}.\end{equation}
Then we have:
$$r_2<r<r^*\;\;\Rightarrow\;\; Z_2(r)<Z_1(r)$$
$$r>r^*\;\;\Rightarrow\;\; Z_1(r)<Z_2(r)$$
\end{lemma}

\begin{proof}
We claim that
\begin{eqnarray}\label{ny}
&&r_1<r<r^*\;\;\Rightarrow\;\; Z_1(r)<\frac{\kappa}{r},\nonumber\\
&&r>r^*\;\;\Rightarrow\;\; Z_1(r)>\frac{\kappa}{r}.
\end{eqnarray}
To see this, consider the function
$\phi(r)=rZ_1(r)$.
Since $Z_1(r)$ is increasing in $r$, $\phi(r)$ is also increasing, and we have $\phi(r_1)=0$ and $\phi(r)\rightarrow \infty$ as
$r\rightarrow \infty$. Therefore there is a unique $\bar{r}>r_1$ satisfying
\begin{equation}\label{bs1}\bar{r}Z_1(\bar{r})=\kappa.\end{equation}
Since $f_1(Z_1(\bar{r}))=Z_1(\bar{r})$ we have
\begin{equation}\label{vv}1- (N_1-v)e^{- \bar{r}Z_1(\bar{r})}-N_2e^{-\sigma \bar{r}Z_1(\bar{r})}- v e^{-\varepsilon \bar{r}Z_1(\bar{r})}=Z_1(\bar{r}).\end{equation}
Substituting \eqref{bs1} in \eqref{vv} gives
\begin{equation*}\label{vv1}1- (N_1-v)e^{- \kappa}-N_2e^{-\sigma \kappa}- v e^{-\varepsilon \kappa}=\frac{\kappa}{\bar{r}}
\;\;\Rightarrow\;\;\bar{r}=\frac{\kappa}{1- (N_1-v)e^{- \kappa}-N_2e^{-\sigma \kappa}- v e^{-\varepsilon \kappa}},\end{equation*}
so we have that $\bar{r}=r^*$ as defined by \eqref{drs}.
We have thus shown that $\phi(r^*)=\kappa$, and by monotonicity of $\phi(r)$, we obtain \eqref{ny}.

Now if $r_2<r<r^*$ then $Z_1(r)<\frac{\kappa}{r}$, so by  Lemma \ref{c2} we have
\begin{equation*}\label{xx1}Z_1(r)=f_1(Z_1(r))>f_2(Z_1(r)),\end{equation*}
which, by Lemma \ref{bp}(iii), implies
$Z_2(r)<Z_1(r)$. Similarly if $r>r^*$ we get 
$Z_2(r)>Z_1(r)$.
\end{proof}

In view of Lemma \ref{ad} we have, defining
\begin{equation*}\label{R01}R_0^*=(N_1+N_2\sigma)r^*=\frac{(N_1+N_2\sigma)\ln(\frac{1}{\alpha})}{1- (N_1-v)\alpha-N_2\alpha^\sigma- v \alpha^\varepsilon},\end{equation*}
that vaccinating the more susceptible group $2$ is preferable if $R_0<R_0^*$, while
vaccinating group $1$ is preferable if $R_0>R_0^*$. 
Using \eqref{xse}, $R_0^*$ can also be written as
$$R_0^*=\frac{(N_1+N_2\sigma)\ln(\frac{1}{\alpha})}{1- N_1\alpha-N_2\alpha^\sigma- v (\alpha^\varepsilon-\alpha)}.$$
The proof of Theorem \ref{main} will be completed by
showing that no mixed vaccination schedule outperforms 
the two possibilities considered above.

\subsection{Optimality of vaccinating one group}
\label{mixed}

We now analyze the optimization problem defined in \eqref{SIop} and show that 
the solution is necessarily a boundary solution: either $(v_1,v_2)=(v,0)$
or $(v_1,v_2)=(0,v)$. We reformulate the problem as a one-dimensional problem by
setting
$$\zeta(v_1)=Z(v_1,v-v_1),$$
and we wish to show that $\zeta(v_1)$ is maximized either at $v_1=0$ or
at $v_1=v$. 
We will in fact prove:

\begin{lemma}
Let $r^*$ be defined by \eqref{drs}. Assume $r>r_2$. Then:

(i) If $r\neq r^*$ then $\zeta'(v_1)\neq 0$ for all
$v_1\in (0,v)$, so that $\zeta(v_1)$ is maximized either at 
$v=0$ or at $v=v_1$.

(ii) If $r=r^*$ then $\zeta(v)=\frac{\kappa}{r^*}$ for all $v\in [0,v]$, that is $\zeta(v)$ a constant function.
\end{lemma}

\begin{proof}
Note that the assumption $r>r_2$  implies that $\zeta(v_1)\in (0,1)$ for all $v_1\in [0,v]$ (herd immunit cannot be achieved).

(i) Let us assume, that $\zeta'(v_1^*)=0$ for some $v_1^*\in (0,v)$, and show that $r=r^*$. 
Substituting $v_2=v-v_1$ in \eqref{SIfsz} we have, for $v_1\in [0,v]$,
\begin{equation}\label{fz1}1- (N_1-v_1)e^{- r\zeta(v_1)}-(N_2-v+v_1)e^{-\sigma r\zeta(v_1)}- v_1 e^{-\varepsilon r\zeta(v_1)}- (v-v_1)e^{-\sigma\varepsilon r\zeta(v_1)}  =\zeta(v_1),\end{equation}
and differentiating this relation with respect to $v_1$ leads to
$$r\left[ (N_1-v_1)e^{- r\zeta(v_1)}+\sigma (N_2-v+v_1)e^{-\sigma r\zeta(v_1)}+\varepsilon v_1 e^{-\varepsilon r\zeta(v_1)}+\sigma\varepsilon  (v-v_1)e^{-\sigma\varepsilon r\zeta(v_1)} \right]\zeta'(v_1)$$
$$+e^{- r\zeta(v_1)}-e^{-\sigma r\zeta(v_1)}-  e^{-\varepsilon r\zeta(v_1)}+e^{-\sigma\varepsilon r\zeta(v_1)}  =\zeta'(v_1).$$
Thus, for $v_1^*$ satisfying $\zeta'(v_1^*)=0$ we obtain
\begin{equation}\label{cr}e^{- r\zeta(v_1^*)}-e^{-\sigma r\zeta(v_1^*)}-  e^{-\varepsilon r\zeta(v_1^*)}+e^{-\sigma\varepsilon r\zeta(v_1^*)}=0.\end{equation}
On the other hand, substituting $v_1=v_1^*$ in \eqref{fz1}, we have
\begin{equation*}\label{fz2}1- (N_1-v_1^*)e^{- r\zeta(v_1^*)}-(N_2-v+v_1^*)e^{-\sigma r\zeta(v_1^*)}- v_1^* e^{-\varepsilon r\zeta(v_1^*)}- (v-v_1^*)e^{-\sigma\varepsilon r\zeta(v_1^*)}  =\zeta(v_1^*),\end{equation*}
which, together with \eqref{cr} gives
\begin{equation}\label{fz3}1- (N_1-v)e^{- r\zeta(v_1^*)}-N_2e^{-\sigma r\zeta(v_1^*)}
-  ve^{-\varepsilon r\zeta(v_1^*)}  =\zeta(v_1^*),\end{equation}
Note that \eqref{cr} is equivalent to
$$g\left(e^{-r\zeta(v_1^*)} \right)=0,$$
where $g(x)$ is defined in \eqref{defg}. Since 
in the proof of Lemma \ref{c2} we obtained that 
$g$ has a unique root $\alpha=e^{-\kappa(\sigma,\varepsilon)}$, 
we see that for $\zeta'(v_1^*)$ to hold, we need
\begin{equation}\label{ne}r\zeta(v_1^*)= \kappa(\sigma,\varepsilon).\end{equation}
Substituting \eqref{ne} into \eqref{fz3} gives
$$1- (N_1-v)e^{- \kappa}-N_2e^{-\sigma \kappa}
-  ve^{-\varepsilon \kappa} =\frac{1}{r}\kappa,\;\;\Rightarrow\;\;
r=\frac{\kappa}{1- (N_1-v)e^{- \kappa}-N_2e^{-\sigma \kappa}
-  ve^{-\varepsilon \kappa}}=r^*,$$
so we have shown that $r=r^*$ is the unique value of 
$r$ of which $\zeta(v_1)$ can have a critical point
$v_1^*\in (0,1)$.

(ii) Now we assume that $r=r^*$ and will show that 
$\zeta(v_1)=\frac{\kappa}{r^*}$ for all $v_1\in (0,v)$.
Indeed, by \eqref{fz1}, $\zeta(v_1)$ is the unique solution of
\begin{equation}\label{fz11}1- (N_1-v_1)e^{- r^*\zeta(v_1)}-(N_2-v+v_1)e^{-\sigma r^*\zeta(v_1)}- v_1 e^{-\varepsilon r^*\zeta(v_1)}- (v-v_1)e^{-\sigma\varepsilon r^*\zeta(v_1)}  =\zeta(v_1),\end{equation}
so it suffices to show that \eqref{fz11} holds identically when we substitute $\zeta(v_1)=\frac{\kappa}{r^*}$, that is to show
\begin{equation}\label{fz111}1- (N_1-v_1)e^{- \kappa}-(N_2-v+v_1)e^{-\sigma \kappa}- v_1 e^{-\varepsilon \kappa}- (v-v_1)e^{-\sigma\varepsilon \kappa}  =\frac{\kappa}{r^*}.\end{equation}
By \eqref{drs} we have
$$\frac{\kappa}{r^*}=1- (N_1-v)e^{- \kappa}-N_2e^{-\sigma \kappa}- v e^{-\varepsilon \kappa},$$
so \eqref{fz111} is equivalent, after some simplification, to 
\begin{equation}\label{fz12}e^{- \kappa}-e^{-\sigma \kappa}-  e^{-\varepsilon \kappa}+e^{-\sigma\varepsilon \kappa}  =0,\end{equation}
which holds due to the fact that $\alpha=e^{-\kappa}$
is the solution of \eqref{xse}. We have therefore proved \eqref{fz11}, finishing the proof of (ii).
\end{proof}

Our proof of Theorem \ref{main} is thus complete.

\section{Dependence of the threshold reproduction number on parameters}
We define the threshold value~$R_{\rm threshold}$ as the minimal value of~$R_0$ at which the relative error in the number of survivors (those not infected),
\begin{equation}\label{eq:SIerror}
E=\frac{Z(R_0,\vec v^*)-Z(R_0,\vec v_{\rm optimal})}{1-Z(R_0,\vec v_{\rm optimal})}
\end{equation} resulting from the asymptotic allocation~$\vec v^*$, see~\eqref{eq:SIoptimalAllocation}, and the computed optimal allocation~$\vec v_{\rm optimal}$ is smaller than 1\%. 

In the two group example considered in Theorem~\ref{main}, the transition value~$R_0^*$ which is closely related to~$R_{\rm threshold}$ is given explicitly by~\eqref{eq:R0*}.  \eqref{eq:R0*} implies that~$R_0^*$ increases both with vaccine efficacy and with vaccine coverage, see also Fig~\ref{fig:R0star}.
In Fig~\ref{fig:RthresholdVariousCountries} we consider more general cases and compute~$R_{\rm threshold}$ for various countries.  As in the example of the USA in the main text (Fig 2), age demographics for each country were taken from the UN World Population Prospects 2019~\cite{UNreport}, using nine age groups, of sizes $N_j$ ($1\leq j\le 9$) corresponding to $10$-year increments, with the last group comprising those of age $80$ and older.  The contact matrices were taken from~\cite{prem2020projecting} and adapted to  10-year increments~\cite{Bubar916}.  In all cases, we used the age dependent relative susceptibility profile for SARS-CoV-19 from~\cite{davies2020age}.  In the USA example presented in Fig 2 of the main text with vaccine of efficacy of 80\% and vaccine coverage of 55\%,~$R_{\rm threshold}\approx 5.7$.  For this choice of parameters, we observe quite small variability of~$R_{\rm threshold}$ among different countries with values ranging from $5$ to~$5.75$ and a mean of~$5.4$.  

Fig~\ref{fig:RthresholdVariousCountries}A presents~$R_{\rm threshold}$ also for vaccine efficacy of 70\% and 90\%.  As expected,~$R_{\rm threshold}$ increases with an increase in vaccine efficacy:  The average~$R_{\rm threshold}$ changes from~$5.4$ at a vaccine efficacy of 80\% to~$R_{\rm threshold}\approx4.3$ when vaccine efficacy is 70\%, and to~$R_{\rm threshold}\approx8$ when vaccine efficacy is 90\%.  We again observe quite small variability of~$R_{\rm threshold}$ among different countries.

We further consider the dependence of~$R_{\rm threshold}$ on vaccine coverage.  We find that the mean value of~$R_{\rm threshold}$ for a range of vaccine coverage values between 40\% to 70\% is~$R_{\rm threshold}\approx6$ where most values lie in the range~$4.25<R_{\rm threshold}<9.5$.  In addition, outliers with $R_{\rm threshold}\approx 12$ are observed, see Fig~\ref{fig:RthresholdVariousCountries}B.  Contrary to the behavior of~$R_0^*$ in the two group example, we find only weak correlation between~$R_{\rm threshold}$ and vaccine coverage values.  This points to the complexity of the solution structure.  


\begin{figure}
\centering
\includegraphics[width=0.8\textwidth]{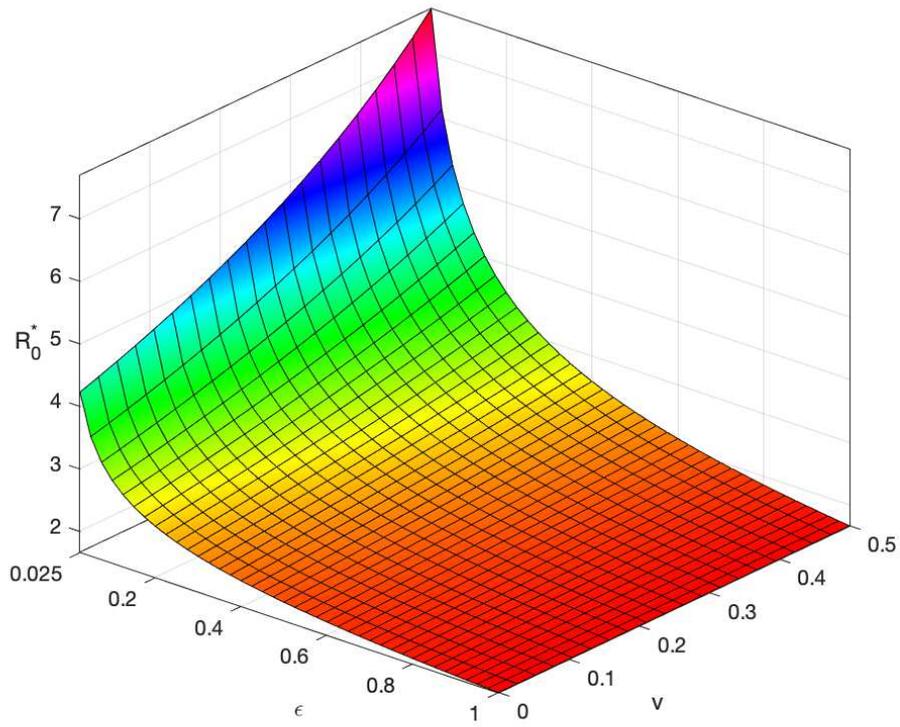}
\caption{{\bf $R_0^*$ as function of vaccine coverage and vaccine efficacy.} Graph of~\eqref{eq:R0*} for~$N_1=N_2=0.5$, $v=0.4$, $\sigma=2$.  Note that the surface is plotted in a regime~$\varepsilon>0.025$ avoiding values closer to zero since~$\lim_{\varepsilon\to0} R_0^*=\infty$}\label{fig:R0star}
\end{figure}

\begin{figure}
\centering
\includegraphics[width=0.8\textwidth]{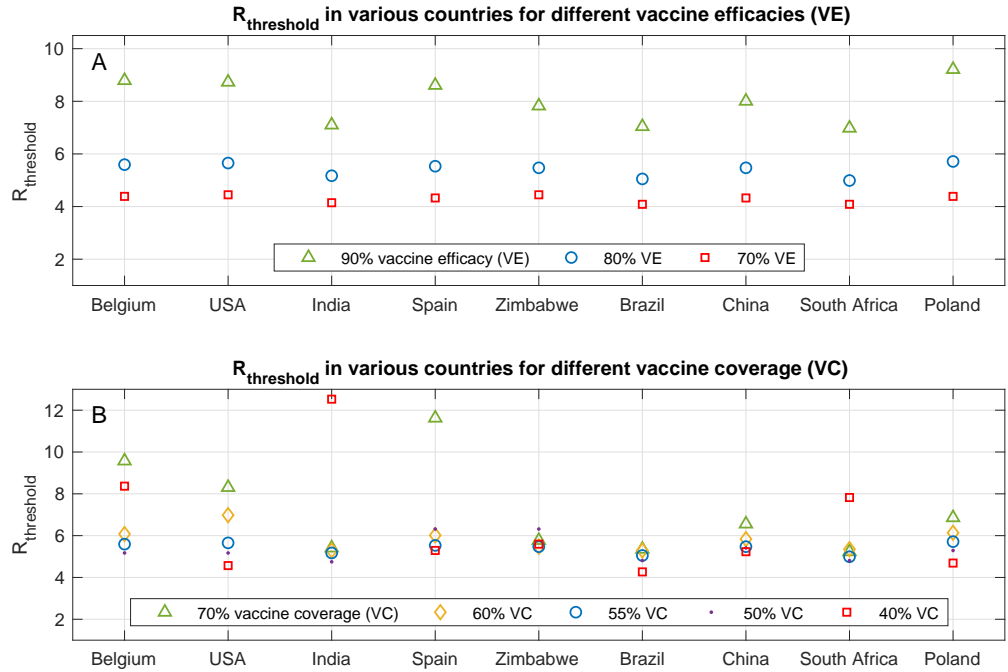}
\caption{The threshold value~$R_{\rm threshold}$ at which the relative error in the number of survivors reduces below 1\% computed for parameters corresponding to various countries and for A: different values of vaccine efficacy.  B: different values of vaccine coverage.}\label{fig:RthresholdVariousCountries}
\end{figure}
\newpage

\end{document}